\documentclass[letterpaper, 10 pt, conference]{ieeeconf}  % Comment this line out
\usepackage{amsmath,amssymb}
\usepackage{tikz}
\usepackage{graphicx}
\usepackage[hidelinks]{hyperref}
\usepackage{amsmath}
\usepackage{dsfont}
\usepackage{circuitikz}
\newtheorem{lem}{Lemma}

\newtheorem{rem}{Remark}
\newtheorem{defn}{Definition}
\newtheorem{cor}{Corollary}
\newtheorem{assum}{Assumption}

\newtheorem{theo}{Theorem}

\newtheorem{pb}{Problem}
\newcommand{\tp}{\intercal}		% transpose
\DeclareMathOperator\VVEC{\textsl{vec}} 
    
\IEEEoverridecommandlockouts                              % This command is only
\title{\LARGE \bf
Decentralized Control Problems with Substitutable Actions
}

\author{Seyed Mohammad Asghari and Ashutosh Nayyar% <-this % stops a space
\thanks{S. M. Asghari and A. Nayyar are with the Department of
Electrical Engineering, University of Southern California, Los Angeles,
CA 90089 USA (e-mail: %asgharip@usc.edu; avestimehr@ee.usc.edu 
        {\tt\small asgharip@usc.edu}; {\tt\small ashutosn@usc.edu}).}%        
 \thanks{This research was supported by NSF under grants ECCS 1509812 and CNS 1446901.}% <-this % stops a space
}

\begin{document}

\maketitle
\thispagestyle{empty}
\pagestyle{empty}

%%%%%%%%%%%%%%%%%%%%%%%%%%%%%%%%%%%%%%%%%%%%%%%%%%%%%%%%%%%%%%%%%%%%%%%%%%%%%%%%
\begin{abstract}

We consider a decentralized system with multiple controllers and define  substitutability of one controller by another in open-loop strategies. We explore the implications of this property on  the optimization of closed-loop strategies. In particular, we focus on the decentralized LQG problem with substitutable actions. Even though the problem we formulate does not belong to the known classes of ``simpler'' decentralized problems such as partially nested or quadratically invariant problems,  our results show that, under the substitutability assumption,  linear strategies are optimal and  we provide a complete state space characterization of optimal strategies. 
 We also identify a family of information structures that all give the same optimal cost as the centralized information structure under the substitutability assumption.
Our results suggest that open-loop substitutability can work as a counterpart of  the information structure requirements that enable simplification of decentralized control problems. 

\end{abstract}

%%%%%%%%%%%%%%%%%%%%%%%%%%%%%%%%%%%%%%%%%%%%%%%%%%%%%%%%%%%%%%%%%%%%%%%%%%%%%%%%
\section{Introduction}\label{sec:intro}

%Points to make: \\
%
%Non partially nested but Linear optimal\\
%
%Substitutability as a counterpart of partial nestedness\\

The difficulty of finding optimal strategies in decentralized control problems has been well-established in the literature \cite{witsenhausen, LipsaMartins:2011b,blondel}. In general, the optimization of strategies can be a non-convex problem over infinite dimensional spaces \cite{YukselBasar:2013}. Even the celebrated linear quadratic Gaussian (LQG) model of centralized control presents difficulties in the decentralized setting \cite{ witsenhausen, LipsaMartins:2011b, mahajan_survey}. There has been significant interest in identifying classes of decentralized control problems that are more tractable.  Information structures of decentralized control problems, which describe what information is available to which controller,  have been closely associated with their tractability. Problems with partially nested \cite{HoChu:1972} or stochastically nested information structures \cite{Yuksel:2009} and problems that satisfy quadratic invariance \cite{RotkowitzLall:2006} or funnel causality \cite{BamiehVoulgaris:2005} properties have been identified as ``simpler'' than the general decentralized control problems. 

%The information structure of a decentralized control problem describes what information is available to which controller. These structures are the central feature of decentralized control problems that determine whether or not the problem belongs to one of the  classes of tractable problems mentioned above.

In this paper, instead of starting from the information structure of the problem, we  first look at open-loop strategies under which controllers take actions without any observations. Stated another way, we start with a trivially simple information structure: no controller knows anything (except, of course, the model of the system and the cost objective). 

%In contrast to open-loop strategies, closed-loop strategies allow controllers to take actions as a function of their observations.  Exactly what observations can be used by controllers (as described by the information structure) determines the tractability or intractability of the control problem.

We define a property of open-loop decentralized control, namely the substitutability of one controller by another, and explore its implications on  optimization of closed-loop strategies (under which controllers take actions as functions of their observations). In particular, we focus on the decentralized LQG problem with substitutable actions. Even though the problem we formulate does not belong to one of the simpler classes mentioned earlier (partially nested, quadratic invariant etc.),  (i) our results show that linear strategies are optimal;  (ii)  we provide a complete state space characterization of optimal strategies; (iii) we also identify a family of information structures that all achieve the same cost as the centralized information structure.
Our results suggest that open-loop substitutability can work as a counterpart of  the information structure requirements that enable simplification of decentralized control problems. 

Our work shares conceptual similarities with the work on internal quadratic variance \cite{lessard_iqi, Lessard:phd} which identified problems that are not quadratically invariant but can still be reduced to (infinite dimensional) convex programs. In contrast to this work, we explicitly identify optimal control strategies.

\subsection{Notation}
Uppercase letters denote random variables/vectors and their corresponding realizations are represented by lowercase letters. Uppercase letters are also used to denote matrices. $\mathrm{E}[\cdot]$ denotes  the expectation of a random variable. When random variable $X$ is normally distributed with mean $\mu$ and variance $\Sigma$, it is shown as $X \sim \mathcal{N}(\mu, \Sigma)$.

For a sequence of column vectors $X, Y, Z,...$, the notation $\textsl{vec}(X,Y,Z,...)$ denotes vector $[X^{\tp}, Y^{\tp}, Z^{\tp},...]^{\tp}$. Furthermore, the vector $\textsl{vec}(X_1, X_2,...,X_t)$ is denoted by $X_{1:t}$. The transpose and Moore-Penrose pseudo-inverse of matrix $A$ are denoted by $A^{\tp}$ and $A^{\dagger}$, respectively. 
The identity matrix and zero vector are denoted by $I$ and 0 respectively and their dimensions are inferred from the context.
% denotes the Moore-Penrose pseudo-inverse of matrix $A$.
%For sub-matrices, which require double indexing, we use superscripts. A matrix can be partitioned as $A = \begin{bmatrix} A^1 &A^2       \end{bmatrix}$ where the dimension of $A^i$ is inferred by the context.
%The notation $0_{n \times m}$ denotes a $n \times m$ matrix with all entries equal to zero and $I_{n \times n}$ denotes a $n \times n$ identity matrix. 

\section{Substitutable Actions}

%%%%General system

%%%%Notion of pairwise substitutability or one sided?

We consider a stochastic system with $n$ controllers.  The  dynamics of the system are given as:
\begin{align}
\label{equation:xx1}
X_{t+1} = f(X_t,U^1_t,\ldots, U^n_t, W_t), \quad t=1,\ldots,T-1.
\end{align}
where $X_t$ is the state of the system at time $t$, $U^i_t$ is the action of controller $i$ at time $t$ and $W_t$ is a random noise variable. The state takes value in the set $\mathcal{X}$, the control action of the $i$th controller takes value in the set $\mathcal{U}^i$ and the noise $W_t$ takes value in the set $\mathcal{W}$. We use $U_t$ to denote the vector $\VVEC(U^1_t,\ldots,U^n_t)$. 

The system operates in discrete
time for a horizon $T$. At time step $t$, the system incurs a cost given as a function of the state and control actions: $c(X_t,U^1_t,\ldots,U^n_t)$. The control objective is to minimize the expected value of the total cost accumulated over the $T$ time steps: $\mathrm{E}[\sum_{t=1}^T c(X_t,U^1_t,\ldots,U^n_t)]$.

We say that controller $i$ is using open-loop control strategy if its control actions are  a function only of time and not of any information (observations) obtained from the system. Otherwise, we say that controller $i$ is using a closed-loop control strategy.
Based on the dynamics and the cost function, we can define a notion of open-loop substitutability among controllers.

\begin{defn}\label{defn:subs_1}
We say that controller 1 can substitute for controller 2 in open-loop control if for every $u^1 \in \mathcal{U}^1, u^2 \in \mathcal{U}^2$,  there exists a control action $v^1 \in \mathcal{U}^1$ for controller $1$  such that for all $x \in \mathcal{X}, w \in \mathcal{W}$, $u^i \in \mathcal{U}^i, i=3,\ldots,n$,
\begin{subequations}
\begin{equation}
f(x,u^1,u^2,u^3,\ldots,u^n,w) = f(x,v^1,0,u^3,\ldots,u^n,w),
\end{equation}
\begin{equation}
\mbox{and}~c(x,u^1,u^2,u^3,\ldots,u^n) = c(x,v^1,0,u^3,\ldots,u^n).
\end{equation}
\end{subequations}
 $v^1$ is a function only of $u^1$ and $u^2$, that is, $v^1=l^{1,2}(u^1,u^2)$, for some function $l^{1,2}$. We will call $l^{1,2}$ the \emph{substitution function} for controller 1 and 2.
\end{defn}
A similar notion of open-loop substitutability can be defined for any pair of controllers in the system. 

If we are considering only open-loop control strategies for all controllers and if controller 1 can substitute for controller 2 in open-loop control (as per Definition \ref{defn:subs_1}), then there is no loss of optimality in fixing all actions of controller 2 to $0$. This is the intuitive meaning of substitutability --- since controller 1 can substitute for controller 2, controller 2 does not need to do anything.  

The open-loop substitutability property has some closed-loop implications. Let us denote by $I^i_t$ the collection of all observations and past control actions that are available to controller $i$ at time $t$. Under a closed-loop strategy, controller $i$ selects its control action as a function of $I^i_t$, that is, 
$U^i_t = g^i_t(I^i_t)$. The collection of functions $ g^i :=\{g^i_1,g^i_2,\ldots,g^i_T\}$ is referred to as controller $i$'s (closed-loop) strategy.

\begin{lem}\label{lem:1}
 Suppose that controller 1 can substitute for controller 2 in open-loop control and that for all time instants $t$, $I^1_t \supseteq I^2_t$, then for any strategies $h^1,h^2,\ldots,h^n$ of $n$ controllers,  there exists strategies $g^1, g^2$ for controllers 1 and 2,  with   $g^2_t(I^2_t) =0$ for all $t$, such that $g^1,g^2,h^3,\ldots,h^n$ achieve the same cost as $h^1,h^2,\ldots,h^n$.
\end{lem}
\begin{proof}
Consider any arbitrary strategies $h^1,h^2,\ldots,h^n$ for the $n$ controllers. 
Define new strategy for controller 1  as follows:
\begin{equation}\label{eq:new_strategy}
g^1_t(I^1_t) = l^{1,2}(h^1_t(I^1_t),h^2_t(I^2_t)),
\end{equation}
where $l^{1,2}$ is the substitution function from Definition \ref{defn:subs_1}. Firstly, note that \eqref{eq:new_strategy} is a valid strategy for controller 1 because $I^2_t \subseteq I^1_t$. If this was not the case, the right hand side of \eqref{eq:new_strategy} would be using information that controller 1 may not have.

The result of the lemma then follows from the observation that the pair  $(U^1_t = g^1_t(I^1_t), U^2_t =0)$ will always have the same effects on dynamics and cost as $(U^1_t=h^1_t(I^1_t), U^2_t=h^2_t(I^2_t))$ because of the substitutability conditions.
\end{proof}

The condition $I^1_t \supseteq I^2_t$ is  necessary for Lemma \ref{lem:1}
to hold. It is easy to construct  examples where $I^1_t$ does not include $I^2_t$ and the second controller cannot be restricted to the ``always zero'' strategy without losing optimality. 

The statement of Lemma \ref{lem:1} can be intuitively interpreted as follows: The open-loop substitutability of controller 2 by controller 1 \emph{and} the fact that controller 1 is better informed make controller 2 essentially redundant for the purpose of cost optimization. 

Lemma \ref{lem:1} suggests that open-loop substitutability, combined with the information structure of the problem, can have implications about the closed-loop problem. In the rest of the paper, we consider a LQG control problem with multiple controllers   and obtain  results much sharper than Lemma \ref{lem:1} for  such problems.

%%%%%%%%%%%%%%%%%

%%%CHANGE TO TWO CONTROLLERS%%%%%%

\section{LQG  problem with state feedback}\label{sec:SF}
\subsection{System Model}\label{sec:sm}

We consider a stochastic system with $n$ controllers where
\begin{enumerate}
\item The state dynamics are given as 
\begin{equation}
\label{dynamics_eq}
X_{t+1} =AX_t + BU_t +W_t, \quad t=1,\ldots,T-1,
\end{equation}
where $X_t, W_t \in \mathbb{R}^{d_x}$ and $U_t \in \mathbb{R}^{d_u}$.
\item The cost at time $t$ is given as
\begin{equation}
\label{cost_eq}
c(X_t,U_t) = (MX_t +NU_t)^{\tp}(MX_t+NU_t).
\end{equation}
%\begin{equation}
%\begin{bmatrix}
%       X_{t}\\[0.3em]
%       U_{t}         \end{bmatrix}^\top
%       \hspace{0mm}
%       \begin{bmatrix}
%       Q & S          \\[0.3em]
%      S^{\tp} & R
%     \end{bmatrix}
%     \hspace{0mm}
%\begin{bmatrix}
%       X_{t}\\[0.3em]
%       U_{t}       \end{bmatrix}
%\end{equation}
%where $ \begin{bmatrix}
%       Q & S          \\[0.3em]
%      S^{\tp} & R
%     \end{bmatrix}$ is a positive semi-definite matrix and $R$ is a positive definite
%matrix.
\item The initial state $X_1$ and the  noise variables $W_t, t=1,\ldots,T-1$ are independent and have Gaussian distributions.
\end{enumerate}
We will make the following assumption about the system.
\begin{assum}
Each controller can substitute for any other controller in open loop control. In other words, for every vector $u =\VVEC(u^1,u^2,\ldots,u^n)$, there exist control actions $v^i = l^i(u)$ for controller $i$, $i=1,\ldots,n$, such that
%\begin{subequations}
\begin{equation}
\label{ass_1}
Bu = B\begin{bmatrix} 0 \\ \vdots \\v^i \\ \vdots \\ 0 \end{bmatrix} ~~\mbox{and}~~Nu = N\begin{bmatrix} 0 \\ \vdots \\v^i \\ \vdots \\  0 \end{bmatrix}.
\end{equation}
%\begin{equation}
%\label{ass_2}
%Nu = N\begin{bmatrix} 0 \\ \vdots \\v^i \\ \vdots \\  0 \end{bmatrix}
%\end{equation}
%\end{subequations}
\end{assum}
\vspace{2mm}
\begin{lem}
\label{lem_3}We can write the $B$ and $N$ matrices in terms of their blocks as 
\[ B = \begin{bmatrix} B^1 &\ldots &B^n\end{bmatrix}, \]
\[ N = \begin{bmatrix} N^1 &\ldots  &N^n\end{bmatrix}. \]

If Assumption 1 is true, then $v^i = \Lambda^i u$ \footnote[1]{
In the previous version of this work \cite{AsghariNayyar:2015}, we stated that if Assumption 1 is true, then $v^i = (B^i)^{\dagger} Bu= (N^i)^{\dagger} Nu$ and $(B^i)^{\dagger} B = (N^i)^{\dagger} N$. This statement is correct only if $B^i$ and $N^i$ are invertible so that $(B^i)^{\dagger} = (B^i)^{-1}$ and $(N^i)^{\dagger} = (N^i)^{-1}$. However, in general, the statement is incorrect. The correct expression for $v^i$ is $v^i = \Lambda^i u$ where $\Lambda^i$ is as given by \eqref{Lamba_eq}.
} satisfies \eqref{ass_1} for $i=1,\ldots,n$ where
\vspace{-2mm}
\begin{align}
\label{Lamba_eq}
\Lambda^i = \begin{bmatrix}
B^i \\
N^i
\end{bmatrix}^{\dagger}
\begin{bmatrix}
B \\
N
\end{bmatrix}.
\end{align}
\end{lem}
\vspace{2mm}
\begin{proof}
The system of equations of \eqref{ass_1} is equivalent to a matrix equation of the form
\begin{align}
\label{subs_condition_eq_system}
P v^i = b, \mbox{where}  
\hspace{2mm}
P = \begin{bmatrix}
B^i \\
N^i
\end{bmatrix}, 
b = \begin{bmatrix}
B\\
N
\end{bmatrix}u.
\end{align}
The general solution of \eqref{subs_condition_eq_system} is $v^i = P^{\dagger}b + (I - P^{\dagger}P)y$ for arbitrary $y$ \cite{ben2006generalized}. By setting $y = 0$, we have $v^i = P^{\dagger}b$.
\end{proof}

An example of a system satisfying Assumption 1 is a two-controller LQG problem where the dynamics and the cost are functions only of the sum of the control actions, that is, ($u_t^1 + u_t^2$). This happens if $B^1 = B^2$ and $N^1= N^2$. In this case, using $v_t^1 = v_t^2 = u_t^1 + u_t^2$ satisfies (\ref{ass_1}), which means that controller 1 can substitute for controller 2 and vice versa.

\subsection{Information Structure}\label{sec:is}
We assume that the state vector $X_t$ consists of $n$ sub-vectors, that is $X_t = \VVEC(X^1_t,X^2_t,\ldots,X^n_t)$. $X^i_t$ can be interpreted as the state of the $i$th sub-system. 
 We assume a local state feedback with perfect recall information structure, that is, the information available to  controller $i$ at time $t$ is:
\begin{align}
\label{equation:xx2}
I_t ^i &= \lbrace X_{1:t} ^i, U_{1:t-1} ^i \rbrace, \hspace{10mm} i=1,2,\ldots,n. 
\end{align}
Controller $i$  chooses action $U_t ^i$ as a function of the information available to it. Specifically, for $i=1,\ldots,n$,
\begin{align}
\label{equation:xx3}
U_t ^i = g_t ^i(I_t ^i), \hspace{10mm} t=1,\ldots,T.
\end{align}
 The collection $g^i = (g^i _1,...,g^i _T)$ is called the control strategy of controller $i$. The performance of the control strategies $g^1$, $g^2, \ldots, g^n$ is measured by the expected cost 
\begin{align}
\label{equation:xx4}
&\mathcal{J}(g^1, \ldots, g^n) \notag \\
 &= \mathrm{E}^{g^1,\ldots,g^n}\left[\sum_{t=1}^{T} 
       (MX_t + NU_t)^{\tp}(MX_t + NU_t)\right]
\end{align}
%\begin{align}
%\label{equation:xx4}
%\mathcal{J}(g^1, \ldots, g^n) = \mathrm{E}^{g^1,\ldots,g^n}\left[\sum_{t=1}^{T} \begin{bmatrix}
%       X_{t}\\[0.3em]
%       U_{t}         \end{bmatrix}^\top
%       \hspace{0mm}
%       \begin{bmatrix}
%       Q & S          \\[0.3em]
%      S^{\tp} & R
%     \end{bmatrix}
%     \hspace{0mm}
%\begin{bmatrix}
%       X_{t}\\[0.3em]
%       U_{t}       \end{bmatrix}\right]
%\end{align}
where the expectation is with respect to the joint probability
distribution on $(X_{1:T}, U_{1:T})$ induced by the choice of $g^1,\ldots,g^n$.

\subsection{Optimal Strategies}

The optimization problem is defined as follows.

\begin{pb}
\label{Prob_1}
For the model described in section \ref{sec:sm} and \ref{sec:is},  find control strategies $g^1,\ldots, g^n$ for the $n$ controllers that minimize the expected cost given by (\ref{equation:xx4}). 
\end{pb}

\begin{rem}
Since we have not imposed any constraints on the matrices $A$ and $B$ in system dynamics,  Problem \ref{Prob_1} may not have partially nested information structure. Thus, we cannot guarantee, at this point, the optimality of  linear control strategies for this problem.
\end{rem}

In addition to the decentralized information structure described above, we will also consider the centralized information structure where all controllers have access to the entire state and action history.

\begin{pb}
\label{Prob_2}
For the model described in section \ref{sec:sm}, assume that  the information available to each controller is,
\begin{align}
\tilde{I}_t = \lbrace X_{1:t}, U_{1:t-1} \rbrace.
\end{align}
Controller $i$ chooses its action according to strategy $h^i=(h^i_1,\ldots,h^i_T)$,
  \[U^i_t = h^i_t(\tilde{I}_t).\]
The objective is to select controller strategies that minimize 
\begin{align}
\label{equation:xx6}
&\mathcal{J}(h^1,\ldots,h^n) \notag \\
  &= \mathrm{E}^{h^1,\ldots,h^n}\left[\sum_{t=1}^{T} 
       (MX_t + NU_t)^{\tp}(MX_t + NU_t)\right]
\end{align}
where the expectation is with respect to the joint probability
distribution on $(X_{1:T}, U_{1:T})$ induced by the choice of $h^1,\ldots,h^n$.
\end{pb}

\subsection{Main results}
In this section, we will show that  we can construct optimal strategies in Problem \ref{Prob_1} from the optimal control strategies of the centralized problem (Problem \ref{Prob_2}). We start with the following observations.

\begin{lem}
\begin{enumerate}
\item The optimal cost in Problem \ref{Prob_2} (with centralized information structure) is a lower bound on the optimal cost in Problem \ref{Prob_1} (with decentralized information structure).
\item The optimal strategies in Problem \ref{Prob_2} are  linear functions of the state, that is, the control vector $U_t = \VVEC(U^1_t,\ldots,U^n_t)$ is given as
\begin{align}
U_t &= K_tX_t = [K^1_t \ldots K^n_t]X_t \notag \\
&= K^1_tX^1_t + \ldots + K^n_tX^n_t.\label{eq:central}
\end{align}
$K_t$ is the centralized gain matrix and $K^i_t$ is its $i$th block.
\end{enumerate}
\end{lem}

We can now state our main result for the state feedback case.

\begin{theo}
\label{theorem_1}
The optimal control strategies in Problem \ref{Prob_1} are given as
\begin{align}
\label{equation:x13}
U_t^i = l_t^i(K^i_t X_t ^i) = \Lambda^i K^i_t X_t ^i, \hspace{5mm} i=1,2,\ldots,n,
\end{align}
where $\Lambda^i$ is given by \eqref{Lamba_eq} and $K^i_t$ is the $i$th block of the centralized gain matrix in \eqref{eq:central}. Moreover, the optimal strategies in Problem \ref{Prob_1} achieve the same cost as the optimal strategies in Problem \ref{Prob_2}.
\end{theo}

\emph{Proof Outline:} Firstly, observe that the strategies given by \eqref{equation:x13} are valid control strategies under the information structure of Problem \ref{Prob_1}.  The optimal control vector $U_t$ under the centralized strategy is a superposition of terms of the form $K^i_tX^i_t$. Note that the term $K^i_tX^i_t$ consists of $n$ sub-vectors (one corresponding to each controller's action).
\[ K^i_t X^i_t = \begin{bmatrix}
K^{i1}_t \\
K^{i2}_t\\
\vdots\\
K^{in}_t
\end{bmatrix}X^i_t.\]
Such a control vector cannot be implemented in the decentralized information structure since it requires each controller to have access to $X^i_t$. 
We now exploit the open loop substitutability of the problem to state that the vector $\textsl{vec}(0, \ldots,l_t^i(K^i_t X_t ^i),\ldots, 0)$ will have the same effect as $K^i_tX^i_t$. This allows us to construct a decentralized strategy with the same performance as the centralized one. We provide a  detailed proof for the more general  case of the output feedback problem in the next section.

We can also derive the following corollary of Theorem \ref{theorem_1}.

\begin{cor}\label{cor:1}
For the model described in section \ref{sec:sm}, consider any information structure under which   the  information of controller $i$ at time $t$, $\hat{I}^i_t$, satisfies 
\[ \lbrace X^i_{t} \rbrace \subseteq \hat{I}^i_t \subseteq \lbrace X_{1:t}, U_{1:t-1} \rbrace,\]
for all $i=1,\ldots,n$ and $t=1,\ldots,T$.
Then, the optimal strategies in this information structure are the same as in Theorem \ref{theorem_1}.
\end{cor}

Corollary \ref{cor:1} identifies memoryless local state feedback as the minimal information structure that  achieves the optimal centralized cost. In other words, it describes the minimal communication and memory requirements for controllers to achieve the optimal centralized cost.

\section{ LQG problem with output feedback}\label{sec:OF}
\subsection{System Model}\label{sec:OF:sm}
We consider the system model described in section \ref{sec:sm} and assume that each controller makes a noisy observation of the system state given as
\begin{align}
\label{equation:x28}
Y_{t}^i = C^{i}X_t + V_{t}^i, \hspace{10mm} i=1,\ldots,n.
\end{align}
 Combining (\ref{equation:x28}) for all controllers gives:
\begin{align}
\label{equation:x281}
Y_{t} = \begin{bmatrix}
       C^1\\[0.3em]
      C^2\\[0.3em]
      \vdots
      \\[0.3em]
      C^n        \end{bmatrix}
X_t + V_{t},
\end{align}
where $Y_t$ denotes $\textsl{vec}(Y_{t} ^1, Y_{t} ^2, \ldots, Y_{t} ^n)$ and $V_t$ denotes $\textsl{vec}(V_{t} ^1, V_{t} ^2, \ldots, V_{t} ^n)$. The initial state $X_1$ and the noise variables $W_t, t=1,\ldots,T-1,$ and $V_t, t=1,\ldots,T-1,$ are mutually independent and jointly Gaussian with the following probability distributions:
\begin{align*}
X_1 \sim \mathcal{N}(0,\Sigma_x), \hspace{5mm} W_t \sim \mathcal{N}(0,\Sigma_w), \hspace{5mm} V_t \sim \mathcal{N}(0,\Sigma_v).
\end{align*}
The information available to the $i$th controller at time $t$ is:
\begin{align}
\label{equation:x31}
I_t ^i &= \lbrace Y_{1:t} ^i, U_{1:t-1} ^i \rbrace \hspace{10mm} i=1,\ldots,n. 
\end{align}
Each controller $i$, chooses  its action $U_t ^i $ according to $U_t^i = g_t^i(I_t^i)$ and the performance of the control strategies of all controllers, ($g^1,\ldots, g^n$), is measured by (\ref{equation:xx4}).

The optimization problem is defined as follows.
\begin{pb}
\label{Prob_3}For the model described above, find control strategies $g^1,\ldots, g^n$ for the $n$ controllers that minimize the expected cost given by (\ref{equation:xx4}). 
\end{pb}

In addition to the decentralized information structure described above, we will also consider the centralized information structure and the corresponding strategy optimization.

\begin{pb}
\label{Prob_4}
For the model described above, assume that the information available to each controller is
\begin{align}
\label{equation:xx5}
\tilde{I}_t = \lbrace Y_{1:t}, U_{1:t-1} \rbrace.
\end{align}
Controller $i$ chooses its action according to strategy $U^i_t = h^i_t(\tilde{I}_t)$. The objective is to select control strategies that minimize (\ref{equation:xx6}).
\end{pb}
The following lemma follows directly from the problem descriptions above and well-known results for the centralized LQG problem with output feedback \cite{kumar_varaiya}. 
\begin{lem} \label{lem_5}
\begin{enumerate}
\item The optimal cost in Problem \ref{Prob_4} (with centralized information structure) is a lower bound on the optimal cost in Problem \ref{Prob_3} (with decentralized information structure).
\item The optimal strategies in Problem \ref{Prob_4} have the form of $U_t = K_t Z_t$ where $Z_t = \mathrm{E} (X_t \vert \tilde{I}_t)$. $Z_t$ evolves according to the following equations:
\begin{align}
\label{equation:x37}
Z_1 &= L_1 Y_1 \nonumber \\
Z_{t+1} &= (I - L_{t+1}C)(AZ_t + BU_t) + L_{t+1}Y_{t+1}. 
%\nonumber \\ u_t &= K_t z_t
\end{align}

We define  $\Sigma_{t} = \mathrm{E} [(X_t - Z_t)(X_t - Z_t)^{\tp} \vert \tilde{I}_t]$ which satisfies the following update equations:
\begin{align}
\label{equation:x381}
\Sigma_1 &= (I - L_{1}C)\Sigma_x \nonumber \\
\Sigma_{t+1} &= (I - L_{t+1}C)(A \Sigma_t A^{\tp} + \Sigma_w).
\end{align}
The matrices $L_1,\ldots, L_T$ in \eqref{equation:x37} and \eqref{equation:x381} satisfy the forward recursion:
\begin{align}
\label{equation:x382}
& L_1 = \Sigma_1 C^{\tp} [C\Sigma_1C^{\tp} +\Sigma_v]^{-1} \nonumber \\
&L_{t+1} = \nonumber \\
& (A \Sigma_t A^{\tp} + \Sigma_{w})C^{\tp}[C(A \Sigma_t A^{\tp} + \Sigma_{w})C^{\tp} +\Sigma_v]^{-1}.
\end{align}

%$K_{1:T}$ satisfies the backward recursion:
%\begin{align}
%\label{equation:x39}
%P_{T} &= \textbf{0} \nonumber \\
%P_t &= A ^{\tp} P_{t+1} A + (A^{\tp} P_{t+1} B + M^{\tp}N)K_t + M^{\tp}M \nonumber \\
%K_t &= - (B^{\tp} P_{t+1} B+ N^{\tp}N)^{-1}(B^{\tp} P_{t+1} A + N^{\tp}M).
%\end{align}
\end{enumerate}
\end{lem}

\subsection{Main results}\label{sec:OF:mr}
In this section, we show that it is possible to construct optimal strategies in Problem \ref{Prob_3} from the optimal control strategy of Problem \ref{Prob_4}. 

\begin{theo}
\label{theorem_2}
Consider Problems \ref{Prob_3} and \ref{Prob_4}, and consider the optimal strategy, $U_t = K_t Z_t$,  of Problem \ref{Prob_4}, where $K_t$ and $Z_t$ are as defined in Lemma \ref{lem_5}. We write $L_{t+1}$ of Lemma \ref{lem_5} as $L_{t+1} = \begin{bmatrix} L_{t+1}^1 &L_{t+1}^2 &\ldots &L_{t+1}^n \end{bmatrix}$. The optimal control strategies of Problem \ref{Prob_3} can be written as
\begin{align}
\label{equation:x46}
U_t ^i =  \Lambda^i K_t S_t^i 
\end{align}
where $\Lambda^i$ is given by \eqref{Lamba_eq} and $S_t ^i$ satisfies the following update equations:
\begin{align}
\label{equation:x41}
       S_{1} ^i &= L_1^i  Y_{1}^i        \nonumber \\
       S_{t+1}^i &= (I- L_{t+1}C)(A S^i_t+ B^iU^i_t)  + L_{t+1}^i Y_{t+1}^i.
\end{align}
Moreover, the optimal strategies in Problem \ref{Prob_3} achieve the same cost as the optimal strategies in Problem \ref{Prob_4}.
\end{theo}
~\\
Observe that the strategies given by (\ref{equation:x46}) and \eqref{equation:x41} are valid control strategies under the information structure of Problem \ref{Prob_3} because they  depend only on $Y_{1:t}^i, U^i_{1:t-1}$ which are included in $I_t^i$. The states $S^i_t$ defined in  \eqref{equation:x41} are related to the centralized estimate $Z_t$ by the following result.
\begin{lem} \label{lem_6}
The centralized state estimate $Z_t$    and the states $S^i_t$ defined in \eqref{equation:x41} satisfy the following equation:
\begin{align}
\label{equation:x47}
Z_t =  \sum_{i=1}^{n}S_t ^i.
\end{align}
%where $S_t^i$ satisfies the recursion given by (\ref{equation:x41}).
\end{lem}

\begin{proof}
We prove the result by induction. For $t=1$, from (\ref{equation:x37}), we have $Z_1 = L_1 Y_1$ and according to (\ref{equation:x41}), 
\begin{align}
\sum_{i=1}^{n}S_1 ^i = L_1^1  Y_{1}^1 + L_1^2  Y_{1}^2 + ... + L_1^n  Y_{1}^n = L_1 Y_1.
\end{align}
Now assume that $Z_t =  \sum_{i=1}^{n}S_t ^i $. We need to show that $Z_{t+1} =  \sum_{i=1}^{n}S_{t+1} ^i $. From (\ref{equation:x37}), it follows that 
\begin{equation}
Z_{t+1} = (I - L_{t+1}C)(AZ_t + BU_t) + L_{t+1}Y_{t+1}.
\end{equation}
 From (\ref{equation:x41}), we have
\begin{align}
\label{equation:x474}
&\sum_{i=1}^{n}S_{t+1} ^i = \sum_{i=1}^{n} [(I- L_{t+1}C)(A S^i_t+ B^i U^i_t)  + L_{t+1}^i Y_{t+1}^i]
\nonumber \\ &=
 (I- L_{t+1}C)(A  \sum_{i=1}^{n}S_t^i + \sum_{i=1}^n B^i U^i_t ) +\sum_{i=1}^{n} L_{t+1}^i Y_{t+1}^i] 
 \nonumber \\ &=
 (I - L_{t+1}C)(AZ_t  + BU_t)+ L_{t+1}Y_{t+1}.
\end{align}
Therefore, $Z_{t+1} =  \sum_{i=1}^{n}S_{t+1} ^i $.
\end{proof}

\begin{rem}
For the case of state feedback, it can be easily shown that $S^i_t = \textsl{vec}(0,\ldots,X^i_t,\ldots,0)$.
\end{rem}

The following result is an immediate consequence of Theorem \ref{theorem_2}.

\begin{cor}\label{cor:2}
For the model described in section \ref{sec:OF:sm}, consider any information structure under which  the  information of controller $i$ at time $t$, $\hat{I}^i_t$, satisfies 
\[ \lbrace Y^i_{1:t}, U^i_{1:t-1} \rbrace \subseteq \hat{I}^i_t \subseteq \lbrace Y_{1:t}, U_{1:t-1} \rbrace,\]
for all $i=1,\ldots,n$ and $t=1,\ldots,T$.
Then, the optimal strategies in this information structure are the same as in Theorem \ref{theorem_2}.
\end{cor}

\subsection{Proof of Theorem \ref{theorem_2}}\label{sec:OF:pot}
For notational conveniences, we will describe the proof for $n=2$. If $U_t = K_t Z_t$ is the optimal control strategy of Problem \ref{Prob_4}, then from Lemma \ref{lem_6}, we have:
\begin{align}
\label{equation:x475}
U_t = K_t Z_t = K_t (S_t ^1 + S_t ^2)
\end{align}
We claim that the decentralized control strategies defined in Theorem \ref{theorem_2}, that is 
\begin{equation}
U_t =  \begin{bmatrix}
           U^1_t\\[0.3em]
           U^2_t
          \end{bmatrix} =\begin{bmatrix}
      \Lambda^1 K_t S_t^1\\[0.3em]
      \Lambda^2 K_t S_t^2        \end{bmatrix},\label{equation:x475a}
\end{equation}
 yield the same expected cost as the optimal centralized control strategies $U_t = K_t Z_t$.
% , that is,
%\begin{align}
%\label{equation:x48}
%&\mathrm{E}\left[\sum_{t=1}^{T} c(X_t, K_t Z_t)\right]= 
%       \mathrm{E}\left[\sum_{t=1}^{T} c\left(X_t, \begin{bmatrix}
%       (N^1)^{\dagger}NK_t S_t^1\\[0.3em]
%      (N^2)^{\dagger}NK_t S_t^2        \end{bmatrix} \right)\right],
%\end{align}
%where $c(x,u) = (Mx+Nu)^{\tp}(Mx+Nu)$.

To establish the above claim, we define cost-to-go functions under the optimal   centralized strategy and the  strategies defined in Theorem \ref{theorem_2}. These functions, denoted by  $\mathbf{V}_r(z)$ and $\hat{\mathbf{V}}_r(z,s^1,s^2)$ for $r=T,T-1,\ldots,1$,  are defined as follows:
\begin{align}
\label{equation:x49}
&\mathbf{V}_r(z) = \notag \\
& \mathrm{E}\Big[\sum_{t=r}^{T} (MX_t + NU_t)^{{\tp}} (MX_t + NU_t) \vert Z_r = z, U_r = K_r z\Big], 
\end{align}
where $U_t$ is given by \eqref{equation:x475} for all $t$, and 
\begin{align}
\label{equation:x50}
&\hat{\mathbf{V}}_r(z,s^1,s^2) = \notag \\
& \mathrm{E}\Big[\sum_{t=r}^{T} (MX_t + NU_t)^{{\tp}} (MX_t + NU_t)  \vert  Z_r =z, S_r^1 = s^1,\nonumber \\ &
 S_r^2 = s^2, U_r^1 = \Lambda^1 K_r s^1, U_r^2 = \Lambda^2 K_r s^2\Big],
\end{align}
where $U_t$ is given by \eqref{equation:x475a} for all $t$.   The function $\hat{\mathbf{V}}_r(z,s^1,s^2)$ in  (\ref{equation:x50}) is defined only for $z= s^1 + s^2$;  $\hat{\mathbf{V}}_r(z,s^1,s^2)$ is undefined for $z\neq s^1 + s^2$.

We will show that for $r=1,...,T$, $\mathbf{V}_r (z) = \hat{\mathbf{V}}_r(z,s^1,s^2) \hspace{2mm} \forall z,s^1,s^2 \in \mathbb{R}^{d_x}$ such that $z=s^1+s^2$. We  follow a backward induction argument. For $r=T$, we have,
\begin{align}
\label{equation:x51}
&\mathbf{V}_T(z) = \mathrm{E}[(MX_T + NU_T)^{{\tp}} (MX_T + NU_T) \vert Z_T = z,  \nonumber \\
&U_T= K_T z] \\
\label{equation:x52}
&\hat{\mathbf{V}}_T(z,s^1,s^2) = \nonumber \\
& \mathrm{E}[(MX_T + NU_T)^{{\tp}} (MX_T + NU_T) \vert  Z_T = 
z, S_T^1 = s^1, \notag \\ &S_T^2 = s^2, U_T^1 = \Lambda^1 K_T s^1, U_T^2 = \Lambda^2 K_T s^2].
\end{align}
Since the only difference between (\ref{equation:x51}) and (\ref{equation:x52}) is with respect to their different control strategies, it suffices to show that the term $NU_T$ is the same under these two control strategies.

Under control action $u_T = K_T z$, we have $Nu_T = NK_T z$.
 Under control actions $u_T^1 = \Lambda^1 K_T s^1, u_T^2 = \Lambda^2 K_T s^2$, we have 
\begin{align}
\label{equation:x521}
Nu_T &= \begin{bmatrix} N^1 &N^2  \end{bmatrix} u_T = N^1 u_T^1 + N^2 u_T^2 \nonumber \\
&= N^1 \Lambda^1 K_T s^1 + N^2 \Lambda^2 K_T s^2 
\end{align}
From the substitutability assumption (Assumption 1) and Lemma \ref{lem_3}, for any vector $u$,  $Nu = N^i l^i(u) = N^i \Lambda^i u$. Therefore, 
\begin{align}
&N^1 \Lambda^1 K_T s^1 = NK_T s^1, \notag \\
&N^2 \Lambda^2 K_T s^2 = NK_T s^2. 
\end{align}
  (\ref{equation:x521}) can now be written as,
\begin{align}
\label{equation:x522}
N^1 \Lambda^1 K_T s^1 + N^2 \Lambda^2 K_T s^2 = N(K_T s^1 + K_T s^2) = NK_T z
\end{align}
where the last equality is true because $z=s^1 +s^2$. Therefore, $\mathbf{V}_{T}(z) = \hat{\mathbf{V}}_{T}(z,s^1,s^2) \hspace{2mm} \forall z,s^1,s^2 \in \mathbb{R}^{d_x}$ such that $z=s^1+s^2$. 

Now, assume that $\mathbf{V}_{k+1}(z) = \hat{\mathbf{V}}_{k+1}(z,s^1,s^2)\hspace{2mm} \forall z,s^1,s^2 \in \mathbb{R}^{d_x}$ such that $z=s^1+s^2$. We need to show that $\mathbf{V}_{k}(z) = \hat{\mathbf{V}}_{k}(z,s^1,s^2)\hspace{2mm} \forall z,s^1,s^2 \in \mathbb{R}^{d_x}$ with $z=s^1+s^2$. For this, note that one can use dynamic programming arguments to write the cost-to-go functions $\mathbf{V}_k$ and $\hat{\mathbf{V}}_k$   in terms of instantaneous cost and the next stage cost-to-go functions:
\begin{align}
\label{equation:x54}
&\mathbf{V}_k(z) = \mathrm{E}[ (MX_t + NU_t)^{{\tp}} (MX_t + NU_t)
 \vert Z_k = z, U_k =
 \nonumber \\
 &
K_k z] + \mathrm{E}[\mathbf{V}_{k+1}(Z_{k+1}) \vert Z_k = z, U_k = K_k z],
\end{align}
and
\begin{align}
\label{equation:x55}
&\hat{\mathbf{V}}_k(z,s^1,s^2) =
\nonumber \\ &
 \mathrm{E}[(MX_t + NU_t)^{{\tp}} (MX_t + NU_t)
 \vert Z_k = z, 
 S_k^1= s^1, 
 \nonumber \\ &
 S_k^2= s^2,
U_k^1 = \Lambda^1 K_k s^1,
 U_k^2 = \Lambda^2 K_k s^2]
\nonumber \\ &+ \mathrm{E}[\hat{\mathbf{V}}_{k+1}
(Z_{k+1},S_{k+1}^1,S_{k+1}^2) \vert 
Z_k = z, S_k^1= s^1, S_k^2= s^2,
\nonumber \\ &
 U_k^1 =  \Lambda^1
K_k s^1, U_k^2 = \Lambda^2 K_k s^2].
\end{align} 
The first expectation on the right hand side of (\ref{equation:x54}) can be shown to be equal to the first expectation on the right hand side of (\ref{equation:x55}) by repeating the arguments used at time $T$.  Using  Lemma \ref{lem_5}, the second expectation on the right hand side of (\ref{equation:x54}) can be written as,
\begin{align}
&\mathrm{E}[\mathbf{V}_{k+1}(Z_{k+1}) \vert Z_k = z, U_k = K_k z] \notag \\
&= 
\mathrm{E}[\mathbf{V}_{k+1}\Big( (I - L_{k+1}C)(AZ_k + BU_k) + \nonumber \\ &~~~~L_{k+1}(CX_{k+1} +V_{k+1})\Big) \vert Z_k = z, U_k = K_k z] \nonumber \\
&=\mathrm{E}[\mathbf{V}_{k+1}\Big( (I - L_{k+1}C)(A + BK_k)z + L_{k+1}(CAX_k +  \nonumber \\
& ~~~~~CBK_kz + CW_k +V_{k+1})\Big) \vert Z_k = z, U_k = K_k z] \notag \\
& = \mathrm{E}[\mathbf{V}_{k+1}\Big((A+
BK_k)z 
 + L_{k+1}CA(X_k - z)  \nonumber \\ &~~~~+ L_{k+1}(CW_k +V_{k+1})\Big) \vert Z_k = z]. \label{eq:ind_0}
\end{align}
Furthermore, because of the induction hypothesis, the second expectation on the right hand side of (\ref{equation:x55}) can be written as,
\begin{align}
&\mathrm{E}[\hat{\mathbf{V}}_{k+1}(Z_{k+1},S_{k+1}^1,S_{k+1}^2) \vert
Z_k = z, S_k^1=s^1, S_k^2 = s^2, 
\nonumber \\
&U_k^1 
= \Lambda^1 K_k s^1, U_k^2
= \Lambda^2 K_k s^2] \notag \\
& =\mathrm{E}[{\mathbf{V}}_{k+1}(Z_{k+1}) \vert
Z_k = z, S_k^1=s^1, S_k^2 = s^2, 
\nonumber \\
&U_k^1 
= \Lambda^1 K_k s^1, U_k^2
= \Lambda^2 K_k s^2]. \label{eq:ind_1}
\end{align}
\eqref{eq:ind_1} can be further written as
\begin{align}
& \mathrm{E}[{\mathbf{V}}_{k+1}\Big((I- L_{k+1}
C)(Az + B^1U^1_k  + B^2 U^2_k) +\notag \\
&L_{k+1} (CAX_k +  CB^1U^1_k + CB^2U^2_k + CW_k +V_{k+1})\Big)  \vert  \nonumber \\
&Z_k = z, S_k^1=s^1, S_k^2 = s^2, U_k^1 = \Lambda^1 K_k s^1, U_k^2 = 
\Lambda^2 K_k s^2]. \label{eq:ind_2}
\end{align}
From the substitutability assumption (Assumption 1) and  Lemma \ref{lem_3}, for any vector $u$,  $Bu = B^i l^i(u) = B^i \Lambda^i u$. Therefore, 
\begin{align}
&B^1 \Lambda^1 K_k s^1 = BK_k s^1, \notag \\
&B^2 \Lambda^2 K_k s^2 = BK_k s^2. 
\end{align}
\eqref{eq:ind_2} can now be written as 
\begin{align}
& \mathrm{E}[{\mathbf{V}}_{k+1}\Big((I- L_{k+1}
C)(Az + BK_ks^1  + BK_ks^2) +\notag \\
&L_{k+1} (CAX_k +  CBK_ks^1 + CBK_ks^2 + CW_k +V_{k+1})\Big)  \vert \notag \\& Z_k = z,  S_k^1=s^1, S_k^2 = s^2, U_k^1 = \Lambda^1 K_k s^1, \notag \\
& U_k^2 = 
\Lambda^2 K_k s^2]  = \mathrm{E}[{\mathbf{V}}_{k+1}\Big((Az + BK_kz) +\notag \\
&L_{k+1} CA(X_k-z) +   L_{k+1}(CW_k +V_{k+1})\Big)  \vert Z_k = z]. \label{eq:ind_3}
\end{align}
 \eqref{eq:ind_3} is the same as \eqref{eq:ind_0}. Therefore, $\mathbf{V}_{k}(z) = \hat{\mathbf{V}}_{k}(z,s^1,s^2) \hspace{2mm} \forall z,s^1,s^2 \in \mathbb{R}^{d_x}$ such that $z=s^1+s^2$.
%\begin{align}
%=\mathrm{E}[\hat{\mathbf{V}}_{k+1}((I- L_{k+1}C)(A + B K_k) (s^1\nonumber \\
%&
%+ s^2) + L_{k+1} (CX_{k+1} +V_{k+1}))
%\vert Z_k = z, S_k^1=s^1, S_k^2 = s^2, 
%\nonumber \\
%&
%U_k^1 = (N^1)^{\dagger}NK_k s^1, U_k^2= 
%(N^2)^{\dagger}NK_k s^2] =\mathrm{E}[\hat{\mathbf{V}}_{k+1}((I- 
%\nonumber \\
%&
%L_{k+1}C)(A + B K_k) (s^1+ s^2) +L_{k+1} 
%(CAX_k + CB^1 (N^1)^{\dagger}N
%\nonumber \\
%&
%K_k s^1+ CB^2 (N^2)^{\dagger}
%NK_k s^2 + CW_k +V_{k+1}))
%\vert Z_k = z, S_k^1=
%\nonumber \\
%&
%s^1, S_k^2 = s^2] 
%\end{align}
%From Lemma \ref{lem_3}, $Bu = B^i l^i(u) = B^i (N^i)^{\dagger}N u$. Considering $u^i = K_k s^i$, (\ref{equation:x57}) can be written as,
%\begin{align}
%\label{equation:x571}
%&\mathrm{E}[\hat{\mathbf{V}}_{k+1}((I- L_{k+1}C)(A + B K_k) (s^1+ s^2) +L_{k+1} (CAX_k
%\nonumber \\
%&
%+ CBK_k (s^1+ s^2) + CW_k +V_{k+1})) \vert Z_k = z, S_k^1=s^1, S_k^2 
%\nonumber \\
%&= s^2] 
%=
%\mathrm{E}[\hat{\mathbf{V}}_{k+1}((A+BK_k)z + L_{k+1}CA(X_k - z) + L_{k+1}
%\nonumber \\
%&
%(CW_k +V_{k+1})) 
% \vert Z_k = z, S_k^1=s^1, S_k^2 = s^2] 
%\end{align}
%Since $\mathbf{V}_{k+1}(z) = \hat{\mathbf{V}}_{k+1}(z,s^1,s^2) \hspace{2mm} \forall z,s^1,s^2 \in \mathbb{R}^{d_x}: z=s^1+s^2$, we can conclude that (\ref{equation:x56}) and (\ref{equation:x57}) are equal. 

 Now, the expected cost under the centralized  control strategy, $U_t =K_tZ_t$, can be written as,
\begin{align}
&\mathrm{E}\left[\sum_{t=1}^{T} c(X_t, K_t Z_t)\right] \notag \\
&=\mathrm{E}\left[\mathrm{E}\left[\sum_{t=1}^{T} c(X_t, K_t Z_t)\Big\vert Z_1, U_1 = K_1Z_1\right]\right] \nonumber \\
&= \mathrm{E} [\mathbf{V}_1(Z_1)],\label{eq:cost1}
  \end{align}
  while the expected cost under the decentralized strategies of Theorem \ref{theorem_2} can be written as
 \begin{align}
&\mathrm{E}\left[\sum_{t=1}^{T} c\left(X_t, \begin{bmatrix}
       \Lambda^1 K_t S_t^1\\[0.3em]
      \Lambda^2 K_t S_t^2        \end{bmatrix}\right)\right] \nonumber \\
&= 
\mathrm{E}\Big[\mathrm{E}\big[\sum_{t=1}^{T} c\left(X_t, \begin{bmatrix}
       \Lambda^1 K_t S_t^1\\[0.3em]
      \Lambda^2 K_t S_t^2        \end{bmatrix}\right)\Big\vert Z_1, S^1_1, S^2_1, U_1^1 =
\nonumber \\
&      
      \Lambda^1 K_1 S_1^1, U_1^2 =
      \Lambda^2 K_1 S_1^2 \big]\Big]     \notag \\
       &= \mathrm{E} [\hat{\mathbf{V}}_1(Z_1,S_1^1,S_1^2)] \label{eq:cost2}
       \end{align}
Because $\mathbf{V}_{1}(z) = \hat{\mathbf{V}}_{1}(z,s^1,s^2) \hspace{2mm} \forall z,s^1,s^2 
$ such that  $z=s^1+s^2$, \eqref{eq:cost1} and
\eqref{eq:cost2} are equal. Thus, the decentralized control strategies of Theorem \ref{theorem_2} achieve the same expected cost as the optimal centralized strategies.

%%%%%%%%%%%%%%%
\section{Concluding Remarks}
We considered a decentralized system with multiple controllers and defined a property called  substitutability of one controller by another in open-loop strategies. For the LQG problem,  our results show that, under the substitutability assumption,  linear strategies are optimal and  we provide a complete state space characterization of optimal strategies.  
Our results suggest that open-loop substitutability can work as a counterpart of  the information structure requirements that enable simplification of decentralized control problems. 
%%%%%%%%%%%%%%%%%%%%%%%%%%%%%%%%%%%%%%%%%%%%%%%%%%%%%%%%%%%%%%%%%%%%%%%%%%%%%%%%%%%%%%%%%%%%%%%%
\bibliographystyle{IEEEtran}
\bibliography{cdc15,collection}

\end{document}